\documentclass[a4paper]{jpconf}
\usepackage{graphicx}
\usepackage{bm,color}% bold math   \textcolor{red} ON
\usepackage{amsthm}
\begin{document}
\newtheorem{theorem}{Theorem}
\newtheorem{corollary}[theorem]{Corollary}
\newtheorem{lemma}[theorem]{Lemma}
\newtheorem{proposition}[theorem]{Proposition}
\newtheorem{remark}[theorem]{Remark}
\newtheorem{definition}{Definition}

\title{Unbounded formulation of the rotation group}

\author{Yoritaka Iwata}

\address{Tokyo Institute of Technology, Tokyo 152-8550, Japan}
\address{Shibaura Institute of Technology, Tokyo 108-8548, Japan}

\ead{iwata$\_$phys@08.alumni.u-tokyo.ac.jp}

\begin{abstract}
The rotation group is formulated based on the abstract $B(X)$-module framework.
Although the infinitesimal generators of rotation group include differential operators, the rotation group has been formulated utilizing the framework of bounded operator algebra. 
In this article, unbounded formulation of rotation group is established.
\end{abstract}

\section{Introduction}
The rotation group is generated by the angular momentum operator (for textbooks, see Refs.~\cite{50weyl,57yamanouchi}).
The angular momentum operator includes a differential operator, as represented by
\begin{equation} \begin{array}{ll}
{\mathcal L} = -i \hbar ({\bf r} \times \nabla),
\end{array} \end{equation}
where $\hbar$ is a real constant called the Dirac constant. 
The appearance of differential operator $\nabla$ in the representation of ${\mathcal L}$ is essential.
The operator $\nabla$ is an unbounded operator for example in a Hilbert space $L^2({\bf R}^3)$, while it must be treated as a bounded operator in terms of establishing an algebraic ring structure.
Furthermore, the operator boundedness is also indispensable for some important formulae such as the Baker-Campbell-Hausdorff formula and the Zassenhaus formula to be valid.
In general, the exponential of unbounded operators cannot be represented by the power series expansion (cf. the Yosida approximation in a typical proof of the Hille-Yosida theorem; e.g., see Ref.~\cite{79tanabe}).

In this article, using $B(X)$-module~\cite{17iwata-2}, the bounded part is extracted from unbounded operators.
The extracted bounded parts are utilized to formulate the rotation group without disregarding the unboundedness of angular momentum algebra.
%The concept of $B(X)$-module is general enough to provide a foundation of the conventional bounded formulation of Lie algebras.
%In other words, by means of $B(X)$-module, the intersection of the Banach algebra (including only bounded operators) and the extracted bounded part of the Lie algebra (including unbounded operators) is shown.

\section{Rotation group}
Let ${\mathbf R}^3$ be the three-dimensional spatial coordinate spanned by the standard orthogonal axes, $x$, $y$ and $z$.
The angular momentum operator ${\mathcal L}$ is considered in $L^2({\mathbf R}^3)$.
The angular momentum operator
\[ \begin{array}{ll}
{\mathcal L} =  ({\mathcal L}_x, {\mathcal L}_y,{\mathcal L}_z)  \vspace{2.5mm} \\
\end{array} \] 
consists of $x$, $y$, and $z$ components
\[  \begin{array}{ll}
\quad {\mathcal L}_x  = -i \hbar ( y \partial_{z} -  z \partial_{y}),    \vspace{2.5mm} \\
\quad {\mathcal L}_y  = -i \hbar ( z \partial_{x} -  x \partial_{z}),   \vspace{2.5mm} \\
\quad {\mathcal L}_z  = -i \hbar ( x \partial_{y} -  y \partial_{x}), 
\end{array} \]
respectively.
The commutation relations
\begin{equation} \label{eqcom}  \begin{array}{ll}
[{\mathcal L}_x, {\mathcal L}_y] = i \hbar {\mathcal L}_z, \quad
[{\mathcal L}_y, {\mathcal L}_z] = i \hbar {\mathcal L}_x, \quad
[{\mathcal L}_z, {\mathcal L}_x] = i \hbar {\mathcal L}_y
\end{array} \end{equation}
are true, where $[{\mathcal L}_i,{\mathcal L}_j]:={\mathcal L}_i {\mathcal L}_j-{\mathcal L}_j {\mathcal L}_i$ denotes a commutator product ($i,j =x,y,z$).
The commutation of angular momentum operators arises from the commutation relations of the canonical quantization
\begin{equation} \label{omt}  \begin{array}{ll}
[x, p_x] = [y, p_y] = [z, p_z] = i \hbar,  \vspace{1.5mm}  \\
\left[y, p_x \right] = \left[y, p_z\right] = \left[z, p_x\right] = \left[z, p_y\right] = \left[x, p_y\right] = \left[x, p_z\right]  = 0.
\end{array} \end{equation}
Indeed, the momentum operator $p=(p_x,p_y,p_z)$ is represented by $p=-i \hbar ( \partial_{x},  \partial_{y},  \partial_{z})$ in quantum mechanics.
It is remarkable that the commutation is always true for the newtonian mechanics;  i.e., $[x, p_x] = [y, p_y] = [z, p_z] =0$ is true in addition to $\left[y, p_x \right] = \left[y, p_z\right] = \left[z, p_x\right] = \left[z, p_y\right] = \left[x, p_y\right] = \left[x, p_z\right]  = 0$.

Let a set of all bounded operators on $L^2({\mathbf R}^3)$ be denoted by $B(L^2({\mathbf R}^3))$.
A set of operators $\{ {\mathcal L}_k; ~ k = x,y,z \}$ or  $\{ i {\mathcal L}_k/\hbar; ~ k = x,y,z \}$ with the commutation relation (\ref{eqcom}) is regarded as the Lie algebra. 
In particular $\{ {\hat \alpha} {\mathcal L}_x + {\hat \beta} {\mathcal L}_y + {\hat \gamma} {\mathcal L}_z; ~ {\hat \alpha}, {\hat \beta}, {\hat \gamma} \in {\mathbf C} \}$ forms a vector space over the complex number field, while $\{ \alpha( i {\mathcal L}_x/\hbar) + \beta (i {\mathcal L}_y/\hbar) + \gamma (i {\mathcal L}_z/\hbar); ~ \alpha, \beta, \gamma \in {\mathbf R} \}$ is a vector space over the real number field.
It is possible to associate the real numbers $\alpha$, $\beta$, and $\gamma$ with the Euler angles (for example, see Ref.~\cite{82landau}).
%%%
The second term of the right hand side of
\begin{equation} \begin{array}{ll} \label{intem}
  ( r_i \partial_{j} ) ( r_k \partial_{l} ) 
  =  r_i (\partial_{j}   r_k) \partial_{l} +  r_i   r_k \partial_{j} \partial_{l} 
\end{array} \end{equation}
disappears as far as the commutator product $[{\mathcal L}_i,{\mathcal L}_j]$ is concerned, where $r_i$ is equal to $i$, and $i,j,k,l =x$, $y$, or $z$ satisfy $i \ne j$ and $k \ne l$.
This fact is a key to justify the algebraic ring structure of $\{ {\mathcal L}_k ; ~ k = x,y,z \}$.
On the other hand, although ${\mathcal L}_k$ is assumed to be bounded on $L^2({\mathbf R}^3)$ in the typical treatment of the Lie algebra, it is not the case for the angular momentum algebra because of the appearance of differential operators in their definitions. 
From a geometric point of view, the range space $R({\mathcal L_k}) \subset L^2({\mathbf R}^3)$ strictly includes the domain space $D({\mathcal L_k})$; i.e., there is no guarantee for any $u \in L^2({\mathbf R}^3)$ and a certain positive $M \in {\mathbf R}$ to satisfy $ \| {\mathcal L_k} u \|_{ L^2({\mathbf R}^3)} \le M \| u \|_{ L^2({\mathbf R}^3)}$.

\section{Results}
\subsection{$B(X)$-module as a set of pre-infinitesimal generators}
Let $t$ and $s$ satisfy $t, s \in [-T, +T]$ for a certain positive real number $T$.
Let $X$ and $U(t,s)$ be a Banach space and a strongly-continuous two-parameter group on $X$ satisfying the semigroup property:
\[ \begin{array}{ll}
U(t,r) U(r,s) = U(t,s), \vspace{1.5mm} \\
U(s,s) = I,
\end{array} \]
where $I$ denotes the identity operator of $X$.
The logarithmic representation of infinitesimal generators is introduced using $U(t,s)$.
Although there is a generalized relativistic version of logarithmic representation \cite{18iwata-2}, it is sufficient here to utilize the original invertible version \cite{17iwata-1}. \vspace{3.5mm} \\
%%%%%%
${\bf Definition.}$ ~
Let $t$ and $s$ satisfy $t, s \in [-T, +T]$ for a certain positive real number $T$.
{\it Let $Y$ and $u_s$ be a subspace of $X$ and an element included in $Y$, respectively. 
For a given evolution operator $U(t,s)$, the pre-infinitesimal generator $A(t):Y \to X$ is a closed operator defined by
\[
A(t) u_s :=  \mathop{\rm wlim}\limits_{h \to 0}  h^{-1} (U(t+h,t) - I) u_s,
\]
where ${\rm wlim}$ means a weak limit.
\\ }
%%%%%%

Let $A(t)$ commute with $U(t,s)$.
The pre-infinitesimal generator is a generalized concept of infinitesimal generator, but it is not necessarily an infinitesimal generator without assuming the dense property of domain space $Y$ in $X$. 
That is, only the exponentiability with a certain ideal domain is valid to the pre-infinitesimal generators. 
The pre-infinitesimal generator $A(t)$ is represented by
\begin{equation} \label{logr} \begin{array}{ll}
A(t) u_s =  (I + \kappa U(s,t) ) \partial_t a(t, s) u_s,  \vspace{2.5mm}  \\
 \partial_t a(t, s) := \partial_t {\rm Log} (U(t, s) + \kappa I),
\end{array} \end{equation}
if $A(t)$ and $U(t,s)$ commute.
In this representation, Log denotes the principal branch of logarithm being defined by the Riesz-Dunford integral, and $\kappa \ne 0$ denotes a certain complex constant.
The representation (\ref{logr}), which is called the logarithmic representation of operators, corresponds to the generalization of the Cole-Hopf transform \cite{18iwata-1}.
The operator $\partial_t a(t,s)$ is the pre-infinitesimal generator, because $\exp[a(t,s)]$ is always well-defined by the boundedness of $a(t,s)$ on $X$.
Note that $\exp[a(t,s)]$ with $\kappa \ne 0$ does not satisfy the semigroup property \cite{17iwata-3}, and $A(t)$ is not necessarily bounded on $X$.

Using the logarithmic representation of generally-unbounded operators in a Banach space $X$, the module over the Banach algebra is introduced based on Ref.~\cite{17iwata-2}.
Much attention is paid to the part $a(t,s)$ of the logarithmic representation. 
Let $U_i(t, s)$ with $i = 1, 2,  \cdots, n$ be evolution operators satisfying the semigroup property.
For a certain $K \in B(X)$, the operators ${\rm Log}(U_i(t, s)+ K)$ are assumed to commute with each other.
A module over the Banach algebra is defined by
\[ \begin{array}{ll}
B_{Lg}(X)  = \{  {\mathcal K}  {\rm Log} (U_i(t, s) + K); \quad  {\mathcal K} \in B_{ab}(X), ~ K \in B(X),~ t,s \in [-T,+T]  \}
\subset B(X),
\end{array} \]
where $B_{ab}(X)$ means a set of bounded operators on $X$, which commute with ${\rm Log} (U_i(t, s) + K)$.
The set $B_{Lg}(X)$ is called the $B(X)$-module in Ref.~\cite{17iwata-2}.
The operator $\partial_t [ {\mathcal K}  {\rm Log} (U_i(t, s) + K) ]$ is the pre-infinitesimal generator of $\exp [{\mathcal K}  {\rm Log} (U_i(t, s) + K)]$.  
Consequently, as suggested by the inclusion relation $B_{Lg}(X) \subset B(X)$, operators $\partial_t a(t,s)$ are the pre-infinitesimal generators if $a(t,s) \in G_{Lg}(X)$ is satisfied.
In the following the operator norms are simply denoted by $\| \cdot \|$ if there is no ambiguity. 
The pre-infinitesimal generator property is examined for products of operators in the next lemma. \\

\begin{lemma}
Let an operator denoted by 
\[
L  {\rm Log} (U(t, s) + K)
\]
be included in $B_{Lg}(X)$, where the evolution operator $U(t,s)$ is generated by $A(t)$, $L$ is an element in $B_{ab}(X)$, and $K$ is an element in $B(X)$.
Let $K$ and $L$ be further assumed to be independent of $t$.
The product of pre-infinitesimal generators, which is represented by
\begin{equation}  \begin{array}{ll}
 L A(t) = (I + KU(s,t) )  \partial_t [ L  {\rm Log} (U(t, s) + K)],
\end{array} \end{equation}
is also the pre-infinitesimal generators in $X$.
\end{lemma}

\begin{proof}
Since $L$ is independent of $t$,
\[
\partial_t [ L {\rm Log} (U(t, s) + K) ]
 =  L \partial_t [ {\rm Log} (U(t, s) + K)]
\]
is true.
The basic calculi using the $t$-independence of $K$ leads to the product of operator $L A(t)$.
It is well-defined by
\[  \begin{array}{ll}
    (I + K U(s,t) )  \partial_t [ L {\rm Log} (U(t, s) + K)]  \vspace{1.5mm}  \\ 
\quad  = L (I + K U(s,t) )  \partial_t [ {\rm Log} (U(t, s) + K)]    \vspace{1.5mm}  \\
\quad  = L (I + K U(s,t) )    (U(t, s) + K)^{-1}  \partial_t U(t, s)   %\vspace{1.5mm}  \\
 = L A(t)
\end{array} \]
under the commutation assumptions, where the relation $\partial_t U(t, s) = A (t)U(t,s)$ is applied.
%\[  \begin{array}{ll}
% (I + KU_i(s,t) )  \partial_t [  {\mathcal K} {\rm Log} (U_i(t, s) + K)]  \vspace{1.5mm}  \\
%  =  {\mathcal K}  K  {\rm Log} (U_i(t, s) + K)  \partial_t[ U_i(s,t) ] 
%   -   \partial_t [  (I + KU_i(s,t) )   {\mathcal K} {\rm Log} (U_i(t, s) + K)]     
%\end{array} \]
Let $t,s \in [-T,+T]$ satisfy $s<t$. 
The pre-infinitesimal generator property of $L A(t)$ is confirmed by
%\[  \begin{array}{ll}
% \left| \int_s^t  (I + KU_i(s,\tau) )  \partial_{\tau} [  {\mathcal K} {\rm Log} (U_i(\tau, s) + K)]  d\tau \right|   \vspace{1.5mm}  \\
% \le  \sup_{\tau \in [s,t]} (I + KU_i(s,\tau) )  \left| \int_s^t   \partial_{\tau} [  {\mathcal K} {\rm Log} (U_i(\tau, s) + K)]  d\tau \right|  \vspace{1.5mm}  \\
%  =  \sup_{\tau \in [s,t]} (I + KU_i(s, \tau) )  \left|     {\mathcal K} {\rm Log} (U_i(t, s) + K) -   {\mathcal K} {\rm Log} (I + K) \right| 
%\end{array} \]
%leading to
\[  \begin{array}{ll}
 \left\| \int_s^t  (I + KU(s,\tau) )  \partial_{\tau} [  L {\rm Log} (U(\tau, s) + K)]  d\tau \right\| \vspace{1.5mm}  \\
  \le      \left\|    (I + KU(s, \sigma) ) \int_s^t  \partial_{\tau} [ L {\rm Log} (U(\tau, s) + K)]  d\tau \right\| \vspace{1.5mm}  \\
  \le   {\displaystyle \sup_{\tau \in [s,t]} } \left\|  (I + KU(s,\tau) )  \right\|   \left\| \int_s^t  \partial_{\tau} [  L {\rm Log} (U(\tau, s) + K)]  d\tau \right\|  \vspace{1.5mm}  \\
 \le \Bigl(1  +     \left\| K \right\|  {\displaystyle \sup_{\tau \in [s,t]} } \left\| U(s,\tau)  \right\|  \Bigr)
   \left\|     L  \right\|
   \left\|   {\rm Log} (U(t, s) + K) -  {\rm Log} (I + K) \right\|,
\end{array} \]
where a certain real number $\sigma \in [s,t]$ is determined by the mean value theorem. 
Consequently, due to the boundedness of $\int_s^t  (I + KU(s,\tau) )  \partial_{\tau} [  L {\rm Log} (U(\tau, s) + K)]  d\tau$ on $X$, $(I + KU(s,t) )  \partial_t [  L {\rm Log} (U(t, s) + K)] $ is confirmed to be the pre-infinitesimal generator in $X$.
\end{proof}

As for the angular momentum operator, the $t$-independent assumption for operators $K$ and $L$ is satisfied.
Note that $t$-independence assumes a kind of commutation relation.
This lemma shows the product-perturbation for the infinitesimal generators of $C^0$-groups under the commutation, although the perturbation has been studied mainly for the sum of operators.
It is remarkable that the self-adjointness of the operator is not required for this lemma.
For the details of conventional bounded sum-perturbation, and the perturbation theory for the self-adjoint operators, see Ref.~\cite{66kato}. \\

\begin{lemma}
Let an operator denoted by 
\[
L  {\rm Log} (U(t, s) + K(t))
\]
be included in $B_{Lg}(X)$, where the evolution operator $U(t,s)$ is generated by $A(t)$, $L$ is an element in $B_{ab}(X)$, and $K(t)$ is an element in $B(X)$.
Let $L$ and $K(t)$ be $t$-independent and $t$-dependent, respectively.
The operators represented by
\begin{equation}  \begin{array}{ll}  \label{prorep}
G(t)  \partial_t [ L  {\rm Log} (U(t, s) + K(t))]
\end{array} \end{equation}
is the pre-infinitesimal generators in $X$, if the operator $G(\tau) \in B(X)$ is strongly continuous with respect to $\tau$ in the interval $[s,t]$. 
\end{lemma}

\begin{proof}
Let $t,s \in [-T,+T]$ satisfy $s<t$. 
The pre-infinitesimal generator property is reduced to the possibility of applying the mean-value theorem.
\[  \begin{array}{ll}
 \left\| \int_s^t  G(\tau)  \partial_{\tau} [   {\rm Log} (U(\tau, s) + K(\tau))]  d\tau \right\|   \vspace{1.5mm}  \\
  \le      \left\|  G(\sigma)  \int_s^t  \partial_{\tau} [  {\rm Log} (U(\tau, s) + K(\tau))]  d\tau \right\|   \vspace{1.5mm}  \\
  \le   {\displaystyle \sup_{\tau \in [s,t]} } \left\|  G(\sigma)  \right\|    \left\| \int_s^t  \partial_{\tau} [   {\rm Log} (U(\tau, s) + K(\tau))]  d\tau \right\|  \vspace{1.5mm}  \\
 \le   {\displaystyle \sup_{\tau \in [s,t]} } \left\|  G(\sigma)  \right\|
   \left\|   {\rm Log} (U(t, s) + K(t)) -  {\rm Log} (I + K(s)) \right\|,
\end{array} \]
where a certain real number $\sigma \in [s,t]$ is determined by the mean value theorem. 
Consequently, $G(t)  \partial_t [ L  {\rm Log} (U(t, s) + K(t))] $ is confirmed to be the pre-infinitesimal generator in $X$.
\end{proof}

Equation (\ref{prorep}) provides one standard form for the representation of operator products in the sense of logarithmic representation. 
As a result, $B(X)$-module is associated with the pre-infinitesimal generator.\\

\subsection{Preparatory Lemmas}
%The contraction mapping argument is applicable with in the bounded framework.
%\begin{lemma} [Banach, Cauchy, Lipshictz, Picard]
%Let $X$ be $L^2({\mathbf R}^3)$.
%Any $F(u)$ satisfying
%\[ \begin{array}{ll}
%\| F(u)-F(v) \| \le C \| u-v \|
%\end{array} \]
%for any $u, ~v \in X$ and $C < \infty$ is an infinitesimal generator in $X$.
%It implies the boundedness of $F(u)$ on $X$.
%\end{lemma}

%\begin{proof}
%$X$ is a Banach space.
%Under the assumption of the Lipschitz continuity:
%\[ \begin{array}{ll}
%\| F(u)-F(v) \| \le C \| u-v \|
%\end{array} \]
%for any $u, ~v \in X$ leads to the fact that
%\[ \begin{array}{ll}
%u(t) = \int_0^{t} F(u(\tau)) d \tau
%\end{array} \]
%satisfies $d u(t)/dt = F(u(t))$.
%Accordingly, if the solution mapping $S$ is defined by
%\[ \begin{array}{ll}
%S v(t) := \int_0^{t} F(v(\tau)) d \tau,
%\end{array} \]
%for $v(t) \in X$, it leads to 
%\[ \begin{array}{ll}
%\| S (u(t) - v(t)) \| = \left\| \int_0^{t} F(u(\tau)) - F(v(\tau)) d \tau \right\|  \vspace{1.5mm} \\
%\quad \le  \int_0^{t} \| F(u(\tau)) - F(v(\tau)) \| d \tau \vspace{1.5mm}  \\
%\quad \le   C |t| \| u(\tau) - v(\tau) \|.
%\end{array} \]
%It follows that $S$ is a contraction for $t$ satisfying $|t| < C^{-1}$. 
%In this situation a fixed point ${\bar u}(t) \in X ~s.t.~ S{\bar u}(t) = {\bar u}(t)$ exists, so that the unique existence of the local solution for $d u(t)/dt = F(u(t))$ is valid.
%\end{proof}

%The concept of pre-infinitesimal generator is defined by ...

%Two lemmas showing the infinitesimal generator properties are proved at first.

The building blocks of rotation groups are introduced as pre-infinitesimal generators in the following lemmas.
%For $r = (x^2+y^2+z^2)^{1/2}$, let $L_r^2({\mathbf R}^3)$ be defined by
%\[  \begin{array}{ll}
%L_r^2({\mathbf R}^3) = \left\{ u(x,y,z) \in L^2({\mathbf R}^3); 
%~ \| u(x,y,z) \|^2_{L_r^2({\mathbf R}^3)}:= \int \int  \int_{{\mathbf R}^3} e^{-i |r|}  |u(x,y,z)|^2 ~ dx ~ dy ~ dz < \infty  \right\}.
% \end{array} \]
%In the same manner, $L_{r_k}^2({\mathbf R})$ is defined by
%\[  \begin{array}{ll}
%L_{r_k}^2({\mathbf R}) = \{ u(r_k) \in L^2({\mathbf R}); 
%~ \| u(r_k) \|^2_{L_{r_k}^2({\mathbf R})}:= \int_{\mathbf R} e^{- i|r_k|}  |u(r_k)|^2 ~ d r_k  < \infty  \},
% \end{array} \]
%where $r_k$ is either $x$, $y$, or $z$.
%The part $e^{- i |r_k|}$ is regarded as a gauge fixed to the coordinate.
%The spaces $L_{r}^2({\mathbf R}^3)$ and $L_{r_k}^2({\mathbf R})$ are Banach spaces equipped with the above norms. 
 \\

\begin{lemma}
Let $r_k$ be either $x$, $y$, or $z$.
Let $I$ be the identity operator of $L^2({\mathbf R}^3)$.
An operator $i r_k I$ is an infinitesimal generator in $L^2({\mathbf R}^3)$.
%In particular the generated group is represented by the power series expansion.
\end{lemma}
%%%
\begin{proof}
For any $w \in {\mathbf C}$, it is possible to define the exponential function by the convergent power series:
$e^{w} = \sum_{j=0}^{\infty} (w)^j/j!$, so that 
\[  \begin{array}{ll}
 e^{i t r_k I} ={\displaystyle \sum_{j=0}^{\infty}} \frac{1}{j!} (i t r_k I)^j 
  \end{array} \]
is well-defined for $t, r_k \in {\mathbf R}$.
This fact is ensured by the boundedness of the identity operator $I$, although $r_k I$ and $i r_k I$ are not bounded operators in $L^2({\mathbf R})$ if the standard $L^2$-norm is equipped.
It is sufficient for $i r_k I$ to be the pre-infinitesimal generator.

For an arbitrary $r_k \in {\mathbf R}$, an operator $i r_k I$ with its domain $L^2({\mathbf R})$ is the infinitesimal generator in $L^2({\mathbf R})$; 
indeed, the spectral set is on the imaginary axis of the complex plane, and the unitary operator is generated as
%\[   \begin{array}{ll}
%\| i x u \|_{L_x^2({\mathbf R})}^2  = \int_{\mathbf R} e^{- i|x|}  |i xu|^2 ~ d x
%\le {\displaystyle \sup_{x \in {\mathbf R}} } (|x|^2 e^{-|x|}) \| u \|_{L^2({\mathbf R})}
%=   {\displaystyle \sup_{x \in {\mathbf R}} } (|x|^2 e^{-|x|}) \| u \|_{L_x^2({\mathbf R})}
%\end{array}  \]
%for any $x \in {\mathbf R}$ and $u \in L_x^2({\mathbf R}) \subset L^2({\mathbf R})$, where the norm equivalence between $L_x^2({\mathbf R})$ and $L^2({\mathbf R})$ is seen by
\[   \begin{array}{ll}
\int |(e^{it r_k I} u)|^2 ~ d r_k 
%= \int \int \int |e^{it |x| I} u |^2 ~ dx ~ dy~ dz %\vspace{2.5mm} \\\quad 
= \int |u|^2 ~ d r_k.
\end{array}  \]
Consequently, the operator $i r_k I$ is treated as an infinitesimal generator in $L^2({\mathbf R})$ and therefore in $L^2({\mathbf R}^3)$. 
\end{proof}

%In the following, the operator ${i  r_k I} $ is identified with $(I + \kappa e^{i  r_k I} ) {\rm Log} (e^{i  r_k I} + \kappa I)$.

%For $r = (x^2+y^2+z^2)^{1/2}$, let $H^1_r({\mathbf R}^3)$ be defined by
%\[ \begin{array}{ll}
%H_r^1({\mathbf R}^3) = \vspace{1.5mm} \\
% \left\{ u(x,y,z) \in H^1({\mathbf R}^3); 
% \| u(x,y,z) \|^2_{H_r^1({\mathbf R}^3)}:= \ \int \int  \int_{{\mathbf R}^3} \sum_{p=0,1} e^{-|r|}  |(u(x,y,z))^{(p)}|^2 ~ dx %~ dy ~ dz  < \infty  \right\}. 
%\end{array} \]
%In the same manner, $H_{r_i}^1({\mathbf R})$ is also defined. 
%This space is utilized to treat the unboundedness of angular momentum operator. \\

\begin{lemma}
For $r_i$ equal to $x$, $y$, or $z$, an operator $\partial_{r_i} $ with its domain $H^1({\mathbf R}^3)$ is an infinitesimal generator in $L^2({\mathbf R}^3)$.
\end{lemma}

\begin{proof}
The operator $\partial_x$ is known as the infinitesimal generator of the first order hyperbolic type partial differential equations.
For a complex number $\lambda$ satisfying ${\rm Re} \lambda > 0$, let us consider a differential equation
\begin{equation} \label{de1} \begin{array}{ll}
 \partial_x  u(x) =  \lambda u(x) - f(x) 
\end{array} \end{equation}
in $L^2({\mathbf R})$, and
\[ \begin{array}{ll}
 u(x) = - \int_x^{\infty}  \exp [\lambda (x-\xi)] f(\xi) d \xi
\end{array} \]
satisfies the equation.
According to the Schwarz inequality,
\[ \begin{array}{ll}
 \int_{-\infty}^{+\infty} | u(x) |^2 dx
  = \int_{-\infty}^{+\infty} | \int_x^{\infty}  \exp [\lambda (x-\xi)] f(\xi) d \xi |^2 dx \vspace{1.5mm}  \\
\qquad  \le  \int_{-\infty}^{+\infty} \left\{ \int_x^{\infty} \exp \left[ \frac{({\rm Re} \lambda) (x-\xi)}{2} \right]  \exp \left[ \frac{({\rm Re} \lambda) (x-\xi)}{2} \right] |f(\xi)| d \xi \right\}^2 dx  \vspace{1.5mm}  \\
\qquad  \le  \int_{-\infty}^{+\infty} \int_x^{\infty} \exp \left[({\rm Re} \lambda) (x-\xi) \right] d \xi ~
  \int_x^{\infty} \exp \left[({\rm Re} \lambda) (x-\xi) \right] |f(\xi)|^2 d \xi ~ dx
\end{array} \]
is obtained, because $|e^{\lambda/2}|^2 = |e^{{\rm Re}\lambda/2}|^2 ~ |e^{i{\rm Im}\lambda/2}|^2 \le e^{{\rm Re}\lambda}$ is valid if ${\rm Re} \lambda > 0$.
Here the equality
\[ \begin{array}{ll}
 \int_x^{\infty} \exp \left[({\rm Re} \lambda) (x-\xi) \right] d \xi
 =  \int_0^{\infty} \exp \left[(-{\rm Re} \lambda) \xi \right] d \xi 
 =  \frac{1}{{\rm Re} \lambda}
\end{array} \]
is positive valued if ${\rm Re} \lambda > 0$.
Its application leads to
\[ \begin{array}{ll}
 \int_{-\infty}^{+\infty} | u(x) |^2 dx
 \le \frac{1}{{\rm Re} \lambda}  \int_{-\infty}^{+\infty}  \int_x^{\infty} \exp \left[({\rm Re} \lambda) (x-\xi) \right] |f(\xi)|^2 d \xi~ dx  \vspace{1.5mm}\\
\qquad \le \frac{1}{{\rm Re} \lambda}  \int_{-\infty}^{+\infty}   \int_{-\infty}^{\xi} \exp \left[({\rm Re} \lambda) (x-\xi) \right] dx  |f(\xi)|^2 d \xi
\end{array} \]
Further application of the equality
\[ \begin{array}{ll}
 \int_{-\infty}^{\xi} \exp \left[({\rm Re} \lambda) (x-\xi) \right] dx
 =  \int_{-\infty}^0 \exp \left[({\rm Re} \lambda) x \right] dx
 =  \frac{1}{{\rm Re} \lambda}
\end{array} \]
results in
\[ \begin{array}{ll}
 \int_{-\infty}^{+\infty} | u(x) |^2 dx
 \le \frac{1}{{\rm Re} \lambda^2}  \int_{-\infty}^{+\infty}  |f(\xi)|^2  d \xi,
\end{array} \]
and therefore
\[ \begin{array}{ll}
\| (\lambda I - \partial_x)^{-1} f \|_{L^2({\mathbf R})} \le  \frac{1}{{\rm Re} \lambda^2} \| f \|_{L^2({\mathbf R})}.
\end{array} \]
That is, for ${\rm Re} \lambda > 0$,
\[ \begin{array}{ll}
\| (\lambda I - \partial_x)^{-1}  \| \le  \frac{1}{{\rm Re} \lambda}
\end{array} \]
is valid. 
The surjective property of $(\lambda I - \partial_k)$ is seen by the unique existence of solution $u(x) \in L^2({\mathbf R})$ for the initial value problem of Eq.~(\ref{de1}).

A semigroup is generated by taking a subset of the complex plane as
\[ \begin{array}{ll}
\Omega = \{ \lambda \in {\mathbf C}; ~ \lambda = \overline{\lambda} \}
\end{array} \]
where $\Omega$ is included in the resolvent set of $\partial_x$.
For $\lambda \in \Omega$, $(\lambda I - \partial_x)^{-1}$ exists, and
\[ \begin{array}{ll}
\| (\lambda I - \partial_x)^{-n}  \| \le  \frac{1}{({\rm Re} \lambda)^n}
\end{array} \]
is obtained.
Consequently, according to the Lumer-Phillips theorem \cite{61lumer,52phillips} for the generation of quasi contraction semigroup, $\partial_x$ with the domain space $H^1({\mathbf R})$ is confirmed to be an infinitesimal generator in $L^2({\mathbf R})$.
The similar argument is valid to $\partial_y$ and  $\partial_z$.
By considering $(x,y,z) \in {\mathbf R}^3$,  $\partial_k$ with $k=x,y,z$ are the infinitesimal generators in $L^2({\mathbf R}^3)$.
\end{proof}

\subsection{Main theorems}
In order to establish $\{ i {\mathcal L}_k/\hbar; ~ k = x,y,z \}$ as the Lie algebra, it is necessary to show 
\[
\pm  i {\mathcal L}_k/\hbar =  \pm (r_i \partial_{r_j} - r_j \partial_{r_i})
\]
as an infinitesimal generator in $L_{r}^2({\mathbf R}^3)$, where $i,j,k =x,y,z$ satisfies $i \ne j \ne k$. 
%%%%%%  
According to Lemmas 3 and 4, both $r_i I$ and $\partial_{r_j}$ are infinitesimal generators in $L_{r}^2({\mathbf R}^3)$.
It follows that, for $t \in {\mathbf R}$, groups $e^{i t r_i I} $ and $e^{t \partial_{r_j}}$ are well-defined in $L_{r}^2({\mathbf R}^3)$. \\

\begin{theorem}
Let $r_i$ be either $x$, $y$, or $z$.
For $i \ne j$, an operator $\pm r_i \partial_{r_j}$ with its domain space $H_{r}^1({\mathbf R}^3)$ is an infinitesimal generator in $L_{r}^2({\mathbf R}^3)$.
%%%
Consequently, the angular momentum operators 
\[  \begin{array}{ll}
\quad \pm i {\mathcal L}_x/\hbar  = \pm ( y \partial_{z} -  z \partial_{y}),    \vspace{2.5mm} \\
\quad \pm i {\mathcal L}_y/\hbar  = \pm  ( z \partial_{x} -  x \partial_{z}),   \vspace{2.5mm} \\
\quad \pm i {\mathcal L}_z/\hbar  = \pm ( x \partial_{y} -  y \partial_{x})
\end{array} \]
are infinitesimal generators in $L_{r}^2({\mathbf R}^3)$.
%In particular, they behave as bounded operators on $L_{r}^2({\mathbf R}^3)$.
\end{theorem}

\begin{proof}
Let $i \ne j$ be satisfied for $i,j =x,y,z$.
Since $e^{\pm(t-s) \partial_{r_j}}$ is well-defined (cf. Lemma 4) with the domain space $H_{r}^1({\mathbf R}^3)$, its logarithmic representation is obtained by
\[ \begin{array}{ll}
%r_i =  ( I + \kappa e^{(s-t) r_i}  )^{-1} \partial_t {\rm Log} (e^{(t-s) r_i} + \kappa I),   \vspace{2.5mm}   \\
\pm \partial_{r_j} =  ( I + \kappa e^{\pm (s-t) \partial_{r_j}}) \partial_t {\rm Log} (e^{ \pm (t-s) \partial_{r_j}} + \kappa I),  
\end{array} \]
where $\kappa \ne 0$ is a certain complex number. 
%The following two points are essential
%\begin{itemize}
%\item the commutation between $r_i I$ and $\partial_{r_j}$ is true if $i \ne j$ is satisfied (cf. Eq.~(\ref{omt}));
%\item $r_i I$ plays a role of bounded operator in $L_{r}^2({\mathbf R}^3)$ (cf. Lemma 3);
%\end{itemize}
%in terms of defining $r_i \partial_{r_j}$ as a pre.
%%%
The product between $i r_i I$ and $ \pm \partial_{r_j}$ is represented by
\[
\pm i r_i \partial_{r_j} 
= i r_i   ( I - \kappa e^{ \pm (s-t) \partial_{r_j}}) \partial_t [ {\rm Log} (e^{ \pm (t-s) \partial_{r_j}} + \kappa I) ].
% =  ( I + \kappa e^{ \pm (s-t) \partial_{r_j}}) \partial_t [ i r_i {\rm Log} (e^{ \pm (t-s) \partial_{r_j}} + \kappa I)],
\]
Using the commutation and $t$-independence of $r_i I$, it leads to the logarithmic representation
\[
\pm  r_i \partial_{r_j} 
%=  r_i   ( I - \kappa e^{ \pm (s-t) \partial_{r_j}}) \partial_t [ {\rm Log} (e^{ \pm (t-s) \partial_{r_j}} + \kappa I) ]
 =  ( I + \kappa e^{ \pm (s-t) \partial_{r_j}}) \partial_t [  r_i {\rm Log} (e^{ \pm (t-s) \partial_{r_j}} + \kappa I)]
\]
without the loss of generality.
The domain space of $r_i I$ is equal to $L^2({\mathbf R}^3)$, as $e^{ i r_i I}$ is represented by the convergent power series in $L_{r}^2({\mathbf R}^3)$.
The half plane $\{ \lambda \in {\mathbf C}; ~ {\rm Re} \lambda > 0 \}$ is included in the resolvent set of $  \pm r_i \partial_{r_j} $ (cf. the proof of Lemma 4). 
Consequently, for $t, r_i \in {\mathbf R}$, the existence of $e^{\pm t r_i \partial_{r_j}}$ directly follows from the confirmed existence of $e^{\pm t \partial_{r_j}}$ (cf. the proof of Lemma 3).
Being equipped with the domain space $H_{r}^1({\mathbf R}^3)$,  $\pm r_i \partial_{r_j}$ is the infinitesimal generator in $L_{r}^2({\mathbf R}^3)$. 

%In the same manner, if an operator $-\partial_k$ with its domain $H^1({\mathbf R}^3)$ is considered, Eq.~(\ref{de1}) is replaced with
%\begin{equation} \label{de2} \begin{array}{ll}
% \partial_x  u(x) =  -\lambda u(x) + f(x) 
%\end{array} \end{equation}
%in $L^2(-\infty,+\infty)$, and
%\[ \begin{array}{ll}
% u(x) =  \int_x^{\infty}  \exp [-\lambda (x-\xi)] f(\xi) d \xi
%\end{array} \]
%satisfies the equation.

The pre-infinitesimal generator property of sum is also understood by the $B(X)$-module property.
The sum  between $r_i \partial_{r_j} $ and $- r_j \partial_{r_i} $ is represented by
\begin{equation} \begin{array}{ll} \label{sumre}
 ( I + \kappa e^{+ (s-t) \partial_{r_j}})  \partial_t  [ r_i ~ {\rm Log} (e^{+ (t-s) \partial_{r_j}} + \kappa I)]
 - 
 ( I + \kappa e^{- (s-t) \partial_{r_i}}) \partial_t  [ r_j ~ {\rm Log} (e^{- (t-s) \partial_{r_i}} + \kappa I)]  \vspace{1.5mm} \\
 =
  ( I + \kappa e^{+ (s-t) ~ \partial_{r_j}})  
   \partial_t  [ r_i ~ {\rm Log} (e^{+ (t-s) \partial_{r_j}} + \kappa I)  -     r_j ~ {\rm Log} (e^{- (t-s) \partial_{r_i}} + \kappa I)]
  \\
\quad - 
 ( \kappa e^{+ (s-t) ~ \partial_{r_j}}  -  \kappa e^{- (s-t) \partial_{r_i}}) \partial_t  [ r_j ~ {\rm Log} (e^{- (t-s) \partial_{r_i}} + \kappa I)]  \vspace{1.5mm} \\
 =
  ( I + \kappa e^{+ (s-t) ~ \partial_{r_j}}) 
   \partial_t   r_i [  {\rm Log} (e^{+ (t-s) \partial_{r_j}} + \kappa I)  -~ {\rm Log} (e^{- (t-s) \partial_{r_i}} + \kappa I) ]  \\
\quad -
  ( I + \kappa e^{+ (s-t) ~ \partial_{r_j}})  
   \partial_t  [ (r_i   -     r_j )~ {\rm Log} (e^{- (t-s) \partial_{r_i}} + \kappa I)]  \\
%%%   
\quad - 
 ( \kappa e^{+ (s-t) ~ \partial_{r_j}}  -  \kappa e^{- (s-t) \partial_{r_i}}) \partial_t  [ r_j ~ {\rm Log} (e^{- (t-s) \partial_{r_i}} + \kappa I)],
\end{array} \end{equation}
where all the three terms in the right hand side are of the form
\[  \begin{array}{ll}
G(t)  \partial_t [ r_i  {\rm Log} (U(t, s) + K(t))]
\end{array} \]
whose pre-infinitesimal generator properties are proved similarly to Lemma 2.
In particular, the first term in the right hand side can be reduced to the above form with $L=1$ and $t$-dependent $K(t)$ (see the proof of Theorem 2 of Ref.~\cite{17iwata-2}), the parts corresponding to $G(t)$ are strongly continuous, and $r_i$ is independent of $t$.  
After having an integral of Eq.~(\ref{sumre}) in terms of $t$, each term is regraded as a bounded operator on $L_{r}^2({\mathbf R}^3)$.
Consequently, for $i \ne j$, the application of Lemma 2 leads to the fact that
\[  \begin{array}{ll}
\pm ( r_i \partial_{r_j} - r_j \partial_{r_i})
\end{array} \]
with its domain space $H_{r}^1({\mathbf R}^3)$ is the infinitesimal generator in $L_{r}^2({\mathbf R}^3)$.
\end{proof}

This lemma shows the unbounded sum-perturbation for infinitesimal generators of $C^0$-groups.
It directly follows from the sum closedness of $B(X)$-module. 
Again it does not require the self-adjointness of the operator.

\begin{theorem}
For $t, s \in [-T, +T]$, let $V_k(t,s)$ with $k=x,y,z$ in $L_{r}^2({\mathbf R}^3)$ be generated by $i {\mathcal L}_k/\hbar$. 
For a certain complex constant $\kappa \ne 0$, the angular momentum operator $\pm i {\mathcal L}_k/\hbar $ with $k=x,y,z$ is represented by the logarithm
\begin{equation} \label{replk}
\pm i {\mathcal L}_k/\hbar = \pm ( I + \kappa V_k(s,t) ) \partial_t  [  {\rm Log} (V_k(t,s) + \kappa I)],
 \end{equation}
and the corresponding evolution operator is expanded by the convergent power series
\begin{equation} \begin{array}{ll} \label{replke}
V_k(t,s) = e^{ {\rm Log} (V_k(t,s) + \kappa I)} - \kappa I  \vspace{1.5mm} \\ \quad 
= {\displaystyle \sum_{n=0}^{\infty}} \frac{1}{n!}  ( {\rm Log} (V_k(t,s) + \kappa I) )^n - \kappa I  \vspace{1.5mm} \\ \quad 
=  (1-\kappa)I + {\displaystyle \sum_{n=1}^{\infty}} \frac{1}{n!}  ( {\rm Log} (V_k(t,s) + \kappa I) )^n 
\end{array} \end{equation}
where ${\rm Log} (V_k(t,s) + \kappa I)$ is bounded on $L_{r}^2({\mathbf R}^3)$, although ${\mathcal L}_k$ is unbounded in $L_{r}^2({\mathbf R}^3)$.
\end{theorem}

\begin{proof}
According to Theorem 5, the group $V_k(t,s)$ with $k=x,y,z$ is generated by the infinitesimal generator $i {\mathcal L}_k/\hbar$ in $L_{r}^2({\mathbf R}^3)$.
This fact leads to the logarithmic representation 
\begin{equation}
i {\mathcal L}_k/\hbar =  ( I + \kappa V_k(s,t) ) \partial_t  [  {\rm Log} (V_k(t,s) + \kappa I)],
 \end{equation}
where $\kappa \ne 0$ is a certain complex constant.
The relation $V_k(t,s) + \kappa I =  e^{ {\rm Log} (V_k(t,s) + \kappa I )}$ admits the power series expansion of $V_k(t,s)$.
\end{proof}

Let us call the representation shown in Eq.~(\ref{replk}) the collective renormalization, in which a detailed  degree of freedom $r_i \partial_{r_j}$ is switched to a collective degree of freedom ${\mathcal L}_k$.
The collective renormalization is introduced for utilizing another degree of freedom instead of the original degree of freedom.
In a more mathematical sense, the collective renormalization plays a role of simplifying the representation.
Equation (\ref{replke}) ensures the validity of convergent power series expansions used in operator algebras even if they include unbounded operators.
%%%

\section{Concluding remark}
The mathematical foundation of rotation group has been demonstrated.
Although the evolution parameter in this article is denoted by $t,s \in [-T,+T]$, it is more likely to be denoted by $\theta, \sigma \in [-\Theta, +\Theta]$, because the evolution parameter in the present case means the rotation angle.

In summary, the product-perturbation theorem for the $C^0$-group of operators is shown by Lemma 1.
An algebraic concept $B(X)$-module is associated with the pre-infinitesimal generator in Lemma 2.
The angular momentum operators are formulated by the logarithmic representation in Theorem 5 as a result of unbounded sum-perturbation.
The bounded aspect of angular momentum operators is presented with respect to the collective renormalization in Theorem 6. 
%%%
In conclusion, the rotation group is formulated by means of the $B(X)$-module.
The present method based on the $B(X)$-module does not require the self-adjointness of the operator, so that it opens up a way to have a full-complex analysis (neither real nor pure-imaginary analysis) for a class of unbounded operators in association with the operator algebra.

\vspace{10mm}

\ack
The author is grateful to Prof. Emeritus Hiroki Tanabe of Osaka University for fruitful comments.
This work was supported by JSPS KAKENHI Grant No. 17K05440.

\vspace{10mm}

%\begin{figure}[h]
%\begin{minipage}{14pc}
%\includegraphics[width=14pc]{name.eps}
%\caption{\label{label}Figure caption for first of two sided figures.}
%\end{minipage}\hspace{2pc}%
%\begin{minipage}{14pc}
%\includegraphics[width=14pc]{name.eps}
%\caption{\label{label}Figure caption for second of two sided figures.}
%\end{minipage} 
%\end{figure}

\end{document}